\documentclass[11pt,a4paper]{article}

\usepackage{xcolor}
\usepackage{amsfonts,amsmath,amssymb,mathtools,amsthm}
\usepackage{graphicx}
\usepackage{stmaryrd}
\usepackage{thmtools}
\DeclareGraphicsExtensions{.eps}

\usepackage[a4paper,top=2cm,bottom=2cm,left=2.5cm,right=2.5cm,marginparwidth=1.75cm]{geometry}

\newcommand\encadremath[1]{\vbox{\hrule\hbox{\vrule\kern8pt
\vbox{\kern8pt \hbox{$\displaystyle #1$}\kern8pt}
\kern8pt\vrule}\hrule}}
\def\enca#1{\vbox{\hrule\hbox{
\vrule\kern8pt\vbox{\kern8pt \hbox{$\displaystyle #1$}
\kern8pt} \kern8pt\vrule}\hrule}}

\newcommand\framefig[1]{
\begin{figure}[bth]
\hrule\hbox{\vrule\kern8pt
\vbox{\kern8pt \vbox{
\begin{center}
{#1}
\end{center}
}\kern8pt}
\kern8pt\vrule}\hrule
\end{figure}
}

\newcommand\figureframex[3]{
\begin{figure}[bth]
\hrule\hbox{\vrule\kern8pt
\vbox{\kern8pt \vbox{
\begin{center}
{\mbox{\epsfxsize=#1.truecm\epsfbox{#2}}}
\end{center}
\caption{#3}
}\kern8pt}
\kern8pt\vrule}\hrule
\end{figure}
}
\newcommand\figureframey[3]{
\begin{figure}[bth]
\hrule\hbox{\vrule\kern8pt
\vbox{\kern8pt \vbox{
\begin{center}
{\mbox{\epsfysize=#1.truecm\epsfbox{#2}}}
\end{center}
\caption{#3}
}\kern8pt}
\kern8pt\vrule}\hrule
\end{figure}
}

\makeatletter
\@addtoreset{equation}{section}
\makeatother
\newtheorem{theorem}{Theorem}[section]

\newtheorem{proposition}{Proposition}[section]
\newtheorem{lemma}{Lemma}[section]
\newtheorem{corollary}{Corollary}[section]

\theoremstyle{definition}
\newtheorem{remark}{Remark}[section]
\newtheorem{definition}{Definition}[section]

\def\br{\begin{remark}\rm\small}
\def\er{\end{remark}}
\def\bt{\begin{theorem}}
\def\et{\end{theorem}}
\def\bd{\begin{definition}}
\def\ed{\end{definition}}
\def\bp{\begin{proposition}}
\def\ep{\end{proposition}}
\def\bl{\begin{lemma}}
\def\el{\end{lemma}}
\def\bc{\begin{corollary}}
\def\ec{\end{corollary}}
\def\beaq{\begin{eqnarray}}
\def\eeaq{\end{eqnarray}}

\theoremstyle{definition}

\newcommand{\be}{\begin{equation}}
\newcommand{\ee}{\end{equation}}
\newcommand{\beq}{\begin{equation}}
\newcommand{\eeq}{\end{equation}}
\newcommand{\bea}{\begin{eqnarray}}
\newcommand{\eea}{\end{eqnarray}}
\newcommand{\beqq}{\begin{equation*}}
\newcommand{\eeqq}{\end{equation*}}
\newcommand{\beaa}{\begin{eqnarray*}}
\newcommand{\eeaa}{\end{eqnarray*}}

\newcommand{\td}{\tilde}
\usepackage{etoolbox}

\appto\appendix{\counterwithin{equation}{section}}

\usepackage[pdftex]{hyperref}
\hypersetup{colorlinks,urlcolor=magenta,citecolor=red,linkcolor=blue,filecolor=black}

\title{\bf{On the geometry of isomonodromic deformations on the torus and the elliptic Calogero-Moser system}}

\date{\vspace{-5ex}}
 
\author{Mohamad Alameddine}

\begin{document}

\maketitle

\begin{center}
\footnotesize{
Universit\'{e} Jean Monnet Saint-\'{E}tienne, CNRS, Institut Camille Jordan UMR 5208, F-42023, Saint-\'{E}tienne, France. 
}
\end{center}
\vspace{1.0cm}
\begin{abstract}
We consider isomonodromic deformations of connections with a simple pole on the torus, motivated by the elliptic version of the sixth Painlevé equation. We establish an extended symmetry, complementing known results. The Calogero-Moser system in its elliptic version is shown to fit nicely in the geometric framework, the extended symplectic two-form is introduced and shown to be closed.
\end{abstract}

 \textbf{\textit{Keywords;}} Isomonodromic deformations, elliptic functions, elliptic Calogero-Moser system, non-autonomous Hamiltonian system. 
\tableofcontents
\newpage

\section{Introduction}

Since its introduction by Riemann in his paper in 1857, the idea of monodromy preserving deformations has known immense developments, the case of regular singularities was solved by L. Schlesinger \cite{schlesinger1912klasse} and L. Fuchs \cite{fuchs1907lineare}, leading to the famous set of equations named after Schlesinger in the Fuchsian case.  These deformations became an extensive field of studies when the sixth Painlev\'{e} equation was derived from a Fuchsian system and the link with the other Painlev\'{e} equations and the Painlev\'{e} property motivated a huge community which contributed \cite{Fuchs,Gambier,Garnier,Painleve,Picard,schlesinger1912klasse} to our advanced actual understanding of the regular case. The topic was then pursued by R. Garnier and K. Okamoto \cite{okamoto1979deformation,Okamoto1986Iso,Okamoto1986} who studied Garnier systems in their scalar version or Schlesinger systems \cite{schlesinger1912klasse} in their matrix form. Ultimately, it was Garnier who extended the results to all other Painlev\'{e} equations obtaining them as completely integrable conditions. The Hamiltonian formulation of the Painlev\'{e} equations was achieved by A.J. Malmquist \cite{Malmquist1922} while the relations with isomonodromic deformations of linear ordinary differential equations with irregular singularities was given by Okamoto \cite{Okamoto1980}. The Lax representation of the Painlevé equations featured several approaches and methods, the revival of the interest in these methods was done by the Japanese school of Jimbo, Miwa and Ueno \cite{JimboMiwaUeno,JimboMiwa} when they considered a huge family of isomonodromic deformations of meromorphic connections defined on the trivial rank-$n$ vector bundle over the Riemann sphere with an arbitrary pole structure of arbitrary order. Their construction inspired methods aiming to understand the nature of this huge family of deformations, of significant importance are; the \textit{isospectral method} developed by the Montr\'eal school \cite{BertolaHarnadHurtubise2022,HarnadHurtubise}, the \textit{confluence of poles method} \cite{Gaiur2023}, and the \textit{isomonodromic oper gauge method} relying on the symplectic reduction of the Hamiltonian system \cite{marchal2024hamiltonianrepresentationisomonodromicdeformations,MarchalAlameddineP1Hierarchy2023}. \\

Another inspiring approach relies on the study of the symplectic structure hidden behind these systems \cite{AtiyahBott}, linking the studies of fundamental groups of Riemann surfaces to the study of the symplectic structure of these moduli spaces of meromorphic connections. A direction adopted by P. Boalch \cite{BoalchThesis,Boalch2001} which lead to deeper understanding of the \textit{Riemann-Hilbert correspondence}. In this paper, the questions that we hope to address are formulated in another framework, or more precisely over a different Riemann surface, namely the torus. In \cite{manin1996sixthpainleveequationuniversal}, Manin constructed a new version of the $P_{VI}$ equation by an elliptic change of coordinates seen as a morphism of bundles, one over the Riemann sphere, the other constructed over the torus. Manin's construction uncovered an unexpected link between Painelvé equations and the \textbf{elliptic Calogero-Moser (CM) system} both of which are integrable isospectral systems \cite{Krichever1981}. This construction has motivated further work on understanding isomonodromy systems over the genus $1$ surface \cite{Takasaki_1999}. It also led to the introduction of Lax pairs that complement this correspondence (see for instance \cite{Bordner_1998}).  \\

 The Hamiltonian built by Manin from the elliptic version contained a dependence on the deformation parameter which turned it into a non-autonomous Hamiltonian, this fact plays a crucial role in the extended symplectic $2-$form that one shall consider when studying the symplectic structure of the phase space. Manin considered then the affine Weyl group symmetries built by Okamoto \cite{Okamoto1986} for this special form of the $P_{VI}$ system. It was through this observation that the link between the Painlevé equations and (CM) system in its elliptic form (i.e. when the potential admits an elliptic form) was established. The two systems are linked immediately via the introduction of a new time parameter $t$ and the exchange $ 2 \pi i \frac{d}{d\tau} \rightarrow \frac{d}{d \tau}$. This provides the Hamiltonian system
\begin{equation}
    \begin{cases}
         \frac{d q}{d t} = & p 
         \\
     \frac{d p}{d t} = &  -   \frac{\partial H (q,p,t)}{\partial t}
    \end{cases}
\end{equation}
with the equation 
\begin{align}
 \frac{d^2 q}{d t^2} = \sum_{a=0}^3 \alpha_a \wp' (q+\omega_a,\tau)
\end{align}
Assigning to the parameters $(\alpha_a)_{0 \leq a \leq 3}$ all the same value of $-g^2/8$, the system is written, using the properties of the Weierstrass function as 
\begin{align}
    \frac{d^2 q}{d t^2}  = -g^2 \wp'(2q,\tau)
\end{align}
which is exactly the two body $(N=2)$ elliptic CM system, we shall detail this link throughout the article. However, note that the isomonodromic nature of the CM model does not depend only on the explicit change of deformation parameters that links the isospectral and isomonodromic approaches. Indeed, both sides admit huge similarities in the simple setting, nevertheless, as proven in \cite{Marchal_2024} in a study on the rational side, the link is non-trivial especially in a more abstract setting. Krichever \cite{Krichever1981} considered the $N-$body elliptic CM system and showed that this system admits an integrable structure encoded in a \textit{Lax pair representation} 
\begin{align}
    \frac{\partial L(z)}{\partial t}  = [L(z),M(z)],
\end{align}
with the Lax matrix and its auxiliary counter-part being matrix-valued functions of a spectral parameter $z$ defined on the torus. Based on the formulation of the Hitchin system \cite{hitchin1987stable}, a geometric construction of isomonodromic systems on a general Riemann surface was developed by Levin and Olshanetsky \cite{ALevin1997PainleveC} who stated that Manin's form may be seen as an isomonodromic system on the torus $T_\tau$. According to this construction, the Darboux coordinates $q_j$ are identified as the moduli of a $SU(n)$ flat bundle above the torus, with $L$ being a Higgs field defined on the bundle as well. \\

This article contributes to the theory of isomonodromic deformations by explaining the geometric construction of isomonodromy systems above the torus, and then showing that the elliptic CM model is an instance of the bigger geometric picture at least locally. The simple relation with the isospectral approach is discussed, as well as the traditional gauge for the CM Lax pair. The paper is organized as follows, in Section \ref{2}, we motivate the theory of isomonodromic deformations above the torus by considering the elliptic $P_{VI}$ equation, the symmetries of this equation are discussed and extended with an additional symmetry. We review the isomonodromy system above the torus in Section \ref{secisodef} and discuss a sufficient simpler case of one simple pole. In Section \ref{3}, we show how the elliptic CM model fits in this geometric construction and discuss an alternative gauge that preserves the double periodicity. We conclude the section by considering the extended symplectic $2-$form attached to the system. This article is self contained and a short review on elliptic functions and their properties is presented in Appendix \ref{B}.

\section{From rational to elliptic Painlevé $VI$ and symmetries} \label{2}
\subsection{Elliptic Painlevé $VI$}
The motivation behind considering isomonodromic deformations on the torus was revived by the work of Manin \cite{manin1996sixthpainleveequationuniversal}. In his article, Manin was inspired by the integral form of $P_{VI}$ evaluated by Fuchs, which in turn is an instance of a more general construction made in \cite{doi:10.1142/9789812830517_0002} to prove the Mordell functional conjecture. In his formalism, Manin proposed the following $P_{VI}$ version
\begin{align}
    (2 \pi i)^2 \frac{d^2 q}{d \tau^2} = \sum_{a=0}^3 \alpha_a \wp' (q+\omega_a,\tau)  
 \end{align}
where $\wp'(z,\tau)$ denotes the derivative in the first variable of the $\wp-$Weiertrass function given by 
\begin{align}
    \wp(z,\tau) = & \frac{1}{z^2} + \sum_{(m,n)\neq (0,0)} \left( \frac{1}{(z+ m+ n \tau)^2} - \frac{1}{(m+n\tau)^2} \right) 
\end{align}
 The torus being presented as the quotient $T_\tau = \mathbb{C}/ (\mathbb{Z} + \tau \mathbb{Z})$, the $(\omega_a)_{0 \leq a \leq 3}$ are assigned the values of the origin and the half periods 
 \begin{align} \label{halfp}
     \omega_0 = 0 && \omega_1 = \frac{1}{2}  && \omega_2 = \frac{1}{2} + \frac{\tau}{2}&& \omega_3 = \frac{\tau}{2}
 \end{align}

Finally, the parameters $(\alpha_a)_{0 \leq a \leq 3}$ map with the usual parameters $(\alpha,-\beta, \gamma,\frac{1}{2} - \delta)$ of the classical sixth Painlevé equation
\begin{align} \label{P6}
    \frac{d^2 y}{d t^2} = & \frac{1}{2} \left( \frac{1}{y} +\frac{1}{y-1} + \frac{1}{y- t}    \right) \left( \frac{dy}{d t} \right)^2  -  \left( \frac{1}{t} +\frac{1}{t-1} + \frac{1}{y- t}    \right) \frac{d y}{d t} \cr
    & + \frac{y (y-1) (y-t)}{t^2 (t-1)^2} \left( \alpha+\beta \frac{t}{y^2} + \gamma \frac{t-1}{(y-1)^2} + \delta \frac{t (t-1)}{(y-t)2} \right)
\end{align}
The elliptic simpler form is reached through a morphism transferring the coordinates from the rational world of the Riemann sphere to the elliptic world of the torus given by 
\begin{align} \label{change}
    \text{elliptic $P_{VI}$}: \, \, \, \, \, (q,\tau) \longleftrightarrow \, \, \, \, \, \left(y =  \frac{\wp (q,\tau) - \wp (\frac{1}{2},\tau)}{ \wp (\frac{\tau}{2},\tau) - \wp (\frac{1}{2},\tau) } , t  =\frac{\wp (\frac{\tau+1}{2},\tau) - \wp (\frac{1}{2},\tau)}{\wp (\frac{\tau}{2},\tau) - \wp (\frac{1}{2},\tau)}  \right) 
\end{align}
The Hamiltonian formulation of the Painlevé equations \cite{Malmquist1922} was extended smoothly to the elliptic case, in particular, the elliptic $P_{VI}$ was shown to be equivalent to the following Hamiltonian system 
\begin{equation}
    \begin{cases}
            2 \pi i \frac{d q}{d \tau} = &  \frac{\partial H (q,p,\tau)}{\partial p} 
            \\
     2 \pi i \frac{d p}{d \tau} = &  -   \frac{\partial H (q,p,\tau)}{\partial q}
    \end{cases}
\end{equation}
with the Hamiltonian 
\begin{align}
    H (q,p,\tau) =  \frac{p^2}{2} - \sum_{a=0}^3 \alpha_a \wp (q+\omega_a,\tau)  
\end{align} 
In addition, Manin transferred the symmetries studied by Okamoto applying them to his setup, let us review these symmetries and complement them by an additional symmetry that relies on the properties of the $\wp-$Weiertrass function which we call \textit{the extended symmetry} of the elliptic $P_{VI}$.

\subsection{Inherited symmetries of the elliptic $P_{VI}$}
The original sixth Painlevé equation holds symmetries described in terms of the affine Weyl group of type $D_4$ studied by Okamoto \cite{Okamoto1986}. Moreover, this equation admits another symmetry under the symmetric group $S_4$ of the affine Dynkin diagram of $D_4$ which is isomorphic to the $F_4$ affine Weyl group. In this section, we shall discuss the symmetries of the elliptic version which are governed by modular transformations extending the results of Manin \cite{manin1996sixthpainleveequationuniversal}. This leads to the conclusion that the lift of the variables from the rational world to the elliptic world $\mathbb{C} \times \mathbb{H}$ changes the symmetries of $P_{VI}$. \\

Take the torus $T_{\tau}$ and take $\omega_a$ given by (\ref{halfp}), the elliptic $P_{VI}$ equation is given by 
\begin{align} \label{EP6}
    (2 \pi i )^2 \frac{d^2 q}{d \tau^2} = \sum_{a=0}^3 \alpha_a \wp'(q + \omega_a,\tau)
\end{align}
The above equation admits the following set of symmetries
\begin{theorem} [Symmetries of the elliptic $P_{VI}$] The elliptic version of $P_{VI}$ admits the following set of symmetries
\begin{enumerate}
    \item \textit{$S_4$ symmetry;} manifested by the permutation of the parameters $\alpha_i$ and induces two types of transformations taking the form 
    \begin{align}
        (q,\tau) \longrightarrow & \left( \frac{q}{c q + \tau} , \frac{a \tau + b}{c \tau + d} \right) \cr
          (q,\tau) \longrightarrow & \left( q+ \omega_\alpha , \tau \right)
    \end{align}
    This symmetry is generated by the permutation of the sections. 
    \item \textit{Landin transform;} which is due to Landin's identity  
    \begin{align}
        \wp' \left(z,\frac{\tau}{2} \right) = \wp' \left(z, \tau \right) + \wp' \left(z+\frac{\tau}{2},\tau\right) 
    \end{align}
    which translates to the equation, when taking parameters $\left( \alpha_0,\alpha_1,\alpha_0,\alpha_1 \right)$, as
    \begin{align}
         \frac{d^2 q}{d \tau ^2} = & \alpha_0 \left( \wp'(q,\tau) +   \wp' (q+\frac{\tau}{2}, \tau) \right) + \alpha_1 \left(  \wp' (q+\frac{1}{2}, \tau) +  \wp' (q+\frac{\tau+1}{2}, \tau)  \right) \cr
          = & \frac{1}{4} \frac{d^2 q(\tau)}{ d(\tau/2)^2} = \alpha_0 \wp'(q,\frac{\tau}{2}) + \alpha_1 \wp'(q +\frac{1}{2}, \frac{\tau}{2})
    \end{align}
    In other words, this creates a bijection between solutions as the converse is also valid in the following way, if $q(\tau)$ is a solution with parameters $\left( \alpha_0,\alpha_1,\alpha_0,\alpha_1 \right)$, then so is $q(2 \tau)$ with parameters $\left( 4 \alpha_0 , 4 \alpha_1,0,0 \right)$, 
    
    \item \textit{Extended symmetry;} manifested by the change of the parameters which can be absorbed into a change of the solution and the periods of the lattice. In particular, taking $(q,\tau)$ as a solution with parameters $\left( \alpha_0,\alpha_1,\alpha_2,\alpha_3 \right)$, then so is $(j q , j \tau)$ with parameters $\left( \frac{\alpha_0}{j^3}, \frac{\alpha_1}{j^3},\frac{\alpha_2}{j^3},\frac{\alpha_3}{j^3} \right)$ for any non-zero complex number $j$. This creates a bijection between solutions since the converse is also true 
    \begin{align}
        (q,\tau, \alpha_i) \longleftrightarrow ( j  q, j \tau, \frac{\alpha_i}{j^3} )
    \end{align}
    One must keep in mind that in order to have a complete symmetry between these sets of solutions, not only $\tau$ is modified but also the other period of the lattice i.e. the half periods $\omega_a$ are modified to $j \omega_a$.
\end{enumerate}
\end{theorem} 
    
\begin{proof}
Let us mention that due to the double periodicity of the elliptic Weiertrass function and its derivative, the elliptic $P_{VI}$ equation admits a symmetry under the transformation $q \longrightarrow q+m+n \tau$, since this transformation does not affect the derivative on the r.h.s. of the equation. We refer to \cite{manin1996sixthpainleveequationuniversal} for a proof of the inherited $S_4$ symmetry. The second set of symmetries follows from direct computation. The proof of the third set is less trivial, however, the idea may be presented as a property of the derivative of the $\wp$ function. Expressing the set of lattice points as $\Lambda$, one may write
\begin{align}
    \wp (z, \Lambda) = & \frac{1}{z^2} + \sum_{  \lambda \in \Lambda - \{ 0\}} \frac{1}{(z+\lambda)^2} - \frac{1}{\lambda^2} \cr
    = & \frac{1}{z^2} + \sum_{  \lambda \in \Lambda - \{ 0\}} \frac{z (2 + \frac{z}{\lambda}) }{(1+ \frac{z}{\lambda})^2} . \frac{1}{\lambda^3}
\end{align}
The derivative is also written as
\begin{align}
    \wp'(z,\lambda) = -2 \sum_{\lambda \in \Lambda} \frac{1}{(z+\lambda)^3}
\end{align}
It is easier to see that both functions satisfy the property
\begin{align}
     \wp (z, \Lambda) = j^2 \wp ( j z,  j \Lambda), \qquad
     \wp' (z, \Lambda) = j^3 \wp' ( j z,  j \Lambda)
\end{align}
This property is a part of a larger set of properties that describes the action of the modular group on these functions. It is straightforward to see that considering the elliptic $P_{VI}$ with $(j q , j \tau)$ induces the required action on the parameters. 
\end{proof}

 \begin{remark}
Hitchin's equation may be obtained, through a combination of the first two sets of symmetries, in a reduced form given by 
\begin{align}
    \frac{d^2 q}{d \tau^2} = \frac{1}{2 \pi^2} \wp'(z, \tau)
\end{align}
Furthermore, the $S_4$ and the extended symmetry translate into a symmetry for the elliptic CM system as one may fix all the parameters to a fixed value, this is not always the case for the symmetry of the Landin's transform.
 \end{remark}

\section{Isomonodromy systems on the torus} \label{secisodef}

It is well known that the Painlevé equations govern the isomonodromic deformations of rational connections (see for instance Chap. $3$ of \cite{FromGaussToPainleve} for a review on the subject). Does this statement admit an extension to the elliptic setting and a construction of isomonodromy systems over the torus? This question was partially answered in \cite{Takasaki_1999} for the simpler case, the answer relied on the isospectral method. This setup however does not extend to more general settings, the general isomonodromy problem was considered for regular singularities in \cite{Levin}. A general formulation for bundles of higher rank was given by Krichever in \cite{Krichever2}. We shall review the geometric picture and show that the elliptic CM is an instance of the general geometry of isomonodromic deformations.  

\subsection{Geometric construction}
Isomonodromic deformation systems are built over a compact Riemann surface $\Sigma$ with a given number of marked points (or simply poles) $(a_k)_{1\leq k\leq s}$, and with a meromorphic connection, defined over the punctured Riemann surface, admitting a given irregular type (i.e. a given local fundamental form or Birkhoff factorization locally diagonalizing the singular part of the connection) at each pole. When $\Sigma$ has non-zero genus $g>0$, the set of poles and associated irregular types is complemented by a set of matrices encoding the behavior of the connection on a basis of homology cycles $(\mathcal{A}_i,\mathcal{B}_i)_{1\leq i\leq g}$. This extended set provides a characterization of the isomorphism classes of connections at stake. Let us adapt this general situation to the torus $T_\tau$ seen as the quotient $\mathbb{C}/ (\mathbb{Z}+ \tau \mathbb{Z})$. Here $\tau\in \mathbb{H}$ is the modular parameter of the torus restricted to the upper half plane $\mathbb{H}$. The starting point is the differential system satisfied by the horizontal sections 
\begin{align}
    \partial_z \Psi(z) = L(z) \Psi(z)
\end{align}
where $z$ is a local coordinate on the torus and $\Psi(z)$ is seen as a horizontal section from the torus a non-trivial rank-$n$ vector bundle defined over a local neighborhood of $z$.\\  

For simplicity and because the elliptic $P_{VI}$ case and the elliptic CM system belong to this subclass, we shall restrict ourselves to the case of \textit{untwisted connections $L(z)$} defined over the genus $1$ surface that are characterized by the following condition: For any pole $a \in T_{\tau}$ of order $r_a$, there exists a local holomorphic gauge transformation $G_a(z):=G_{a,0}+ G_{a,1}(z-a)+\underset{k=2}{\overset{+\infty}{\sum}}G_{a,k}(z-a)^k$ seen as a formal bundle automorphism that acts on the system 
\begin{align}
    \partial_z \Psi (z) = L(z) \Psi (z) 
\end{align}
by the gauge transformation $\Psi_a(z) := G_a(z) \Psi(z)$, and produces a local diagonalization of the singular part of the connection at $z=a$:
 \small{\beqq \Psi_a(z)\overset{z\to a}{=}\Psi_a^{(\text{reg})}(z)\exp\left(\sum_{k=1}^{r_a}\text{diag}\left(t_{a^{(1)},k},\dots, t_{a^{(n)},k}\right)(z-a)^{-k} + \text{diag}\left(t_{a^{(1)},0},\dots, t_{a^{(n)},0}\right)\ln(z-a)  \right)\eeqq} 
  \normalsize{}
where $\Psi_a^{(\text{reg})}(z)$ is locally holomorphic around $z=a$. This local fundamental decomposition may also be interpreted in terms of the connection as
\beq \label{DiagoCondition}
L_a(z):= G_a \, L (z) dz\, G_a^{-1} + (\partial_z G_a)G_a^{-1}dz = dQ_a(z_a) + \Lambda_a \frac{dz_a}{z_a} + O\left(dz_a\right)
 \,\,\text{where}\,\, Q_a(z_a) = \sum_{k=0}^{r_a-1} \frac{Q_{a,k}}{z_a^{\, k}}
\eeq 
where $z_a:=z-a$ is the local coordinate, $Q_a(z_a)$ and $\Lambda_a$ are diagonal matrices that are respectively called the \textit{irregular type} of $L(z)$ at $a$ and its \textit{residue (also called exponent of formal monodromy)} \cite{Birkhoff}. Note that the location of poles, their irregular type and monodromies are local data at each pole that do not take into account the genus of the underlying Riemann surface $\Sigma$. For connections defined on the Riemann sphere, these data would provide all the parameters for deformations. On the contrary, when considering isomonodromic deformations on the torus, one must also take into account the analytic continuation on the punctured torus of the horizontal sections around the cycles $ \mathcal{A} (z \rightarrow z+1)$ and $\mathcal{B} (z \rightarrow z+\tau)$. In addition to the cycles $\mathcal{C}_{a_1},...,\mathcal{C}_{a_s}$ associated around each pole, this provides additional \textit{monodromy matrices} $M_1$, $M_\tau$, $M_{a_1},\dots,M_{a_s}$ corresponding to
\begin{align} \label{periodicity}
\Psi(z+1)= & \Psi(z) M_1 \cr
\Psi(z+\tau)= & \Psi(z) M_\tau \cr
\Psi(ze^{2i\pi}  -a_k)= & \Psi(z) M_{a_k}
\end{align}  
and satisfying the relation
\begin{align}
    M_1M_\tau M_1^{-1}M_\tau^{-1}=\prod_{k=1}^s M_{a_k}
\end{align} 
By setting some additional conditions on the eigenvalues of the residue matrix, the monodromy around the poles can be set to
\begin{align}
   M_{a_k} \sim e^{2 \pi i \Lambda_a}    
\end{align}
where the equivalence indicates the conjugacy class of the monodromy matrix. By the conjugation action, one may always normalize the horizontal sections diagonalizing the monodromy matrices around the $\mathcal{A}$ cycle. \\

The monodromy matrices and the local fundamental decompositions fix the space of connections that we are considering. Moreover, they provide the full set of deformation parameters (complemented by the module of the torus $\tau$) and the set of monodromies that are preserved by these deformations. Therefore, let us define first the set of connections having the local formal decomposition as above over the punctured torus

\begin{align}
  F_{\mathcal{R},\mathbf{r}} := \left\{ L(z) \in \mathfrak{gl}_n(\mathbb{C}  \,\,   | \,\, L(z) \,\, \text{has irregular type} \,\,  Q_a(z_a) \,\, \text{and local monodromies} \,\, \Lambda_a   \right\} / GL_n(\mathbb{C})
\end{align}
This space does not define a moduli space since it does not take into account the additional monodromy matrices associated to the cycles. By taking into account the additional monodromy related to the homology cycles of the torus, we define the \textit{moduli space of connections}.

\begin{definition} \label{moduli}
    Let us define the moduli space of connections considered 
    \begin{align}
     \hat{\mathcal{M}}_{\tau,\mathcal{R},Q_\mathbf{a},M_\mathbf{a},M_1,M_\tau}  := & \{ L(z) \in F_{\mathcal{R},\mathbf{r}}  \,\,\,  | \,\,\, L(z) \,\, \text{has a global monodromy} \,\, M_1,M_\tau \cr 
     & \text{around the homology cycles of the torus}  \} / GL_n(\mathbb{C}) 
 \end{align}
 where the set $\mathcal{R}$ is the set of poles considered and the equivalence relation indicates the isomorphism classes i.e. two connections are isomorphic when they are related by a gauge transform (one of them lives in the orbit of the other). This space may be viewed as the symplectic reduction of the moduli space of flat $GL_n(\mathbb{C})-$connections defined on the punctured torus \cite{Alekseev1995}, and it has a dimension
 \begin{align}
     dim \,  \hat{\mathcal{M}}_{\tau,\mathcal{R},Q_\mathbf{a},M_\mathbf{a},M_1,M_\tau} = s n^2  + s n+s +1
 \end{align}
 Furthermore, the monodromy matrices define the character variety of the surface 
\begin{align}
    \mathcal{M}_{1,s} = \left\{ (M_1,M_\tau,M_{a_1},...,M_{a_s}) \in GL_n(\mathbb{C}) \,\,\,\,\,\, \text{with} \,\,\,\,\, M_1^{-1} M_\tau^{-1} M_1 M_\tau = \prod_{k=1}^s M_k  \right\} / \sim
\end{align}
The character variety is a space of dimension
\begin{align}
    dim  \,  \mathcal{M}_{1, s} = n(s n +1)  
\end{align}
\end{definition}
We shall now focus on a particular explicit case of this construction fixing only one simple pole at the origin, this case is actually the one which governs the elliptic CM model as we shall see. The compatibility of the system defined with the Lax pair provides the evolutions of the defined Darboux coordinates. Let us be more explicit on this matter in the next section when we consider the simpler case.

 \subsection{The case of one pole}
 Let us consider an instance of the general definition and fix the torus as above with one marked pole of order $r_0=1$ located at the origin of the torus. Since one may use the translations on the torus that do not alter the symplectic structure (the location of the pole would provide a "trivial time" in the sense of \cite{marchal2024hamiltonianrepresentationisomonodromicdeformations}), we may choose to set the pole at $z=0$. In a local neighborhood of the pole, one may define the following space of connections admitting a local behavior at the pole 
\begin{definition}[Space of local connections] Let us define the following set of connections over the torus admitting a simple pole at $z=0$ of order $r_0=1$
\begin{align} \label{def}
    F_{0,1} := \{ L(z):= L^{[0,1]}z^{-1} +  L^{[0,0]} + O(z) \,\,\,\,\, | \, \, \, L^{[0,1]}, L^{[0,0]} \in \mathfrak{gl}_n(\mathbb{C})   \}/ GL_n (\mathbb{C})
\end{align}
where the quotient indicates the action of the reductive complex group, and the orders of the connection are seen as matrices defined on a rank-$n$ non-trivial $GL_n (\mathbb{C})-$bundle.
\end{definition}
\begin{remark} \label{sl}
   One may consider connections with a vanishing trace on the $SL_n(\mathbb{C})-$bundle, we choose to keep the trace at this level. This does not cause any difference when making the link with the elliptic CM model later. 
\end{remark}

By the geometric construction, one may take a representative of a respective orbit without changing the underlying structure and consider this subset of connections with an irregular type defined in the neighborhood of the pole whose entries define the deformation parameters. Note that the set of deformation parameters is complemented by the module of the torus $\tau$. Let us also define as above the horizontal sections solutions of the following differential system
\begin{align} \label{isodef}
    \partial_z \Psi(z) =  L(z) \Psi(z) \qquad
    2 \pi i \partial_\mathbf{t} \Psi(z)  = - A(z) \Psi(z) 
\end{align}
where the additional $A(z)$ matrix is the associated auxiliary matrix completing the Lax pair $(L(z),A(z))$ and $\mathbf{t}$ is the set of deformation parameters. In this particular gauge, the matrix does not have only a simple pole, but rather a set of additional apparent singularities as we will show after introducing the Lax pair of the elliptic CM system. This is due to the Riemann-Roch theorem which forces the introduction of some additional twists when considering only one pole. We shall show the explicit bundle automorphism (or gauge transformation) that links both sides of the problem.   \\
When considering isomonodromic deformations on the torus as discussed, one has the additional cycles $\mathcal{A} (z \rightarrow z+1)$ and $\mathcal{B} (z \rightarrow z+\tau)$. In this specific case, considering a singular marked point located at $z=0$ (to which we assign the cycle $\mathcal{C}_0$), we end up with three monodromy matrices $(M_1,M_\tau,M_0)$. The first equation of (\ref{isodef}) is seen as an ordinary differential equation on the torus that encodes the change of the section in terms of the connection, while the second one ensures the independence of the monodromy matrices of the deformation parameters as is proven in proposition \ref{indepdence}. The periodicity of the horizontal sections in (\ref{periodicity}) provides a Lax pair that is also periodic with respect to the Lattice periods
\begin{align}
    L(z+1) = L(z) \qquad L(z+ \tau) = L(z) \cr
     A(z+1) = A(z) \qquad A(z+ \tau) = A(z)
\end{align}
However, as we shall see, the standard Lax pair of the elliptic CM system is expressed in a different gauge or equivalently on a different bundle linked to the original one via a bundle automorphism. To match the gauge of the system, we shall perform the following gauge transformation
\begin{align} \label{Gx}
    \td{\Psi}(z) := &  G_{l}(z)  \Psi(z) \cr
    := & diag (x(q_1,z),..., x(q_n,z)) \Psi(z)
\end{align}
where the functions $x(u,z)$ are the Lamé functions defined in (\ref{choice}), having properties explained in \ref{PProperties}, of half of the Darboux coordinates $(q_i)_{1 \leq i \leq n}$. The other half of the Darboux coordinates is defined through the entries of the Lax matrix as
\begin{align}
    p_j := \td{L}[j,j]  
\end{align}
and they also satisfy $p_j= 2 \pi i \partial_{\mathbf{t}} q_j$. The quasi-periodic properties of the Lamé functions induces a quasi-periodic behavior for the sections around the cycles and the simple pole. In terms of the monodromy matrices $(M_1,M_\tau,M_0)$, the sections have the following monodromy
\begin{equation}
   \td{ \Psi}(ze^{2 \pi i}) = \td{\Psi}(z) M_0, \qquad \td{\Psi} (z+1) = \td{\Psi} (z) M_1, \qquad \td{\Psi}(z+\tau) = e^{2 \pi i Q} \td{\Psi}(z) M_\tau
\end{equation}
This monodromy translates through (\ref{isodef}) to the Lax pair, we have
\begin{align} \label{action}
    \td{L}(z+1) = & \td{L}(z), \, \, \, \td{A}(z+1) = \td{A}(z) \cr
     \td{L}(z+\tau) = & e^{2 \pi i Q} \td{L}(z) e^{-2 \pi i Q} \cr
     \td{A} (z+ \tau ) = & e^{2 \pi i Q} (  \td{A}(z)+ 2 \pi i \td{L}(z) )  e^{-2 \pi i Q} - 2 \pi i P
\end{align}
where $Q := diag(q_1,...,q_n)$ and $P := 2 \pi i \partial_\mathbf{t} Q$. The above construction ensures the following proposition
\begin{proposition} \label{indepdence}
The monodromy matrices hold no dependence on the parameters $\mathbf{t}$, in other words, we have
\begin{align}
    \frac{d M_0}{d \mathbf{t}} = \frac{d M_1}{d \mathbf{t}} = \frac{d M_\tau}{d \mathbf{t}} = 0
\end{align}
\end{proposition}
\begin{proof}
    Starting from the second part of (\ref{isodef}), one may write
    \begin{align}
        \td{A} (z) =  - 2 \pi i \partial_\mathbf{t} \td{\Psi}(z) \td{\Psi} (z)^{-1} 
    \end{align}
    Let us change the local coordinate and use the monodromy structure of the sections to get the following set of equations 
    \begin{align}
       \td{ A}(z e^{2 \pi i }) = & \td{A} (z)  - 2 \pi i \td{\Psi}(z) \frac{d M_0}{d \mathbf{t}} M_0^{-1}  \td{\Psi} (z)^{-1} \cr
        \td{A} (z+1) = & \td{A}(z)  -  2 \pi i \td{\Psi}(z)  \frac{d M_1}{d \mathbf{t}} M_1^{-1}  \td{\Psi} (z)^{-1} \cr
        \td{A}(z+\tau) = & e^{2 \pi i Q} (  \td{A}(z) +2 \pi i \td{L}(z) )  e^{-2 \pi i Q} - 2 \pi i P   - 2 \pi i \td{\Psi}(z)  \frac{d M_\tau}{d \mathbf{t}} M_\tau^{-1}  \td{\Psi} (z)^{-1} \cr
    \end{align}
    where one uses the relation $2 \pi i \partial_{\mathbf{t}} Q = P$. One sees immediately that the independence of the monodromy matrices on the deformation parameters $\mathbf{t}$ is forced. This result is easily generalisable to the case with several poles.
\end{proof}
The locally defined space $F_{0,1}$ of connections does not hold any information concerning the cycles of the torus and how may the connection transform under these cycles. However, the monodromy matrices above are governed by the character variety of the surface defined in this specific case as follows

\begin{definition} \label{CV}
The space of connections defined is described by the following space 
\begin{align}
     \td{\mathcal{M}}_{\tau ,\{ 0 \},Q_0,M_0,M_1,M_\tau} = & \{ L(z) \in F_{0,1}  \,\,\,  | \,\,\, L(z) \,\, \text{has a global monodromy} \,\, M_1,M_\tau   \cr 
     & \text{around the homology cycles of the torus} \} / GL_n(\mathbb{C})
\end{align}
where the equivalence is understood in terms of the gauge considered under the conjugation action of the reductive complex Lie group, this definition is an instance of the general definition made in def. \ref{moduli}.  The dimension is given by
\begin{align}
    dim \, \td{\mathcal{M}}_{1,1} = n^2+n+1
\end{align}
The character variety of these $GL_n(\mathbb{C})$ connections is defined by the monodromy matrices modulo the conjugation action of the group which identifies the monodromy matrices related by a gauge transformation. This variety takes the form of a cubic surface defined by 
\begin{align}
    \mathcal{M}_{1,1} = \{ M_1, M_\tau, M_0 \in GL_n(\mathbb{C}); \, \, \, \, \, \, \, \,  M_1^{-1}M_\tau^{-1} M_1 M_\tau =M_0 \} / \sim
\end{align}
It is a symplectic variety with a monodromy data satisfying the Goldman bracket. 
\end{definition}

From (\ref{isodef}), one writes the zero curvature equation associated to the system and satisfied by the Lax pair
\begin{align}
      2 \pi i \partial_\mathbf{t} \td{L}(z) + \partial_z \td{A}(z) = [\td{L}(z),\td{A}(z)]
\end{align}
which shall provide the evolution of the Darboux coordinates and their respective Hamiltonian. We shall see in the next section that the elliptic CM model Lax pair indeed fits in the above description, making it an isomonodromic system.

\begin{remark}
    One huge difference between the rational and the elliptic worlds is the choice of the Darboux coordinates, the $P_{VI}$ case corresponds to a case with a genus $g=1$, when considering isomonodromic deformations on the sphere, the symplectic fibre in \cite{marchal2024hamiltonianrepresentationisomonodromicdeformations} becomes of dimension $2$ and thus contains only one set of coordinates and one Hamiltonian describing the system. The situation in the elliptic setting is different, one ends up with $n$ pairs of Darboux coordinates due to the natural choice of the cycles of the torus, and still manages to encode the structure in one Hamiltonian corresponding to the genus of the torus, we shall see this in details when making the identification with the elliptic CM model.   
\end{remark}

\section{The elliptic CM model as an isomondromic system} \label{3}
\subsection{The Lax pair}
Let us define the Lax pair taken from the theory of Hamiltonian reduction of the CM system in its elliptic form
\begin{definition}
Let us define the reduced Lax pair as 
\begin{align} \label{Laxpair}
    \td{L}(z) = & P + i g \sum_{j \neq k} x(q_j-q_k , z) E_{jk} \cr
    \td{A} (z) = & D + i g \sum_{j \neq k} y(q_j-q_k , z) E_{jk}
\end{align}
where 
\begin{align}
    D = & i g \, \text{diag} \left(\sum_{k \neq 1} \wp (q_1-q_k,\tau),..., \sum_{k \neq n} \wp (q_n-q_k,\tau) \right)
\end{align}
the set $(q_j,p_j)_{1 \leq j \leq n}$ are the Darboux coordinates attached to our system, and $E_{jk}$ is the unit matrix with $1$ in the $[j,k]$ entry. This Lax pair has been studied from many perspectives and is obtained through the theory of Hamiltonian reduction (see for instance \cite{etingof2009lecturescalogeromosersystems}).     \\
The function $y(u,z)$ is defined as the $u-$derivative of the function $x(u,z)$, both of them must satisfy the following set of quasi-periodic equations
\begin{align} \label{period}
    x(u,z) y(v,z)- y(u,z) x(v,z) =& x(u+v,z) (\wp (u,\tau) - \wp(v,\tau)) \cr
    x(u,z) y(-u,z) - y(u,z) x(-u,z) =&  \wp'(u,\tau) \cr
     x(u,z) x(-u,z)  = & \wp (z,\tau) - \wp(u,\tau)
\end{align}
This set of equations is essential in order for the Lax pair to satisfy the Lax equation \cite{Krichever1981}. 
\end{definition}

\begin{remark}
    One may refer to the defined Lax pair as a reduced Lax pair since one may consider a more general version of this pair and study the symplectic structure of the Poisson space before the reduction. This is the standard terminology following the study of moduli spaces of meromorphic connections defined over Riemann surfaces. However, this remains beyond the scope of this article and we leave this aspect of our studies for future investigations.
\end{remark}
\begin{remark}
    As discussed in remark \ref{sl}, one may consider the isomonodromic deformations of connections defined on the $SL_n(\mathbb{C})-$bundle over the torus by setting the condition $\sum_j p_j = 0$, this choice is manifested by taking the center of the mass frame as the origin of the torus and is seen through the connection by vanishing the trace of its leading order.  
\end{remark}

\subsection{As an isomonodromy system }
We shall in this section show that the Lax pair of the elliptic $\mathfrak{su}(n)$ CM system fits in the geometric construction of isomonnodromy systems of meromorphic connections above the torus. However, note that this system is an instance of a more general family of Hamiltonian systems defined for any simple Lie algebra. Given a Lie algebra $\mathfrak{g}$, Olshanetsky and Perelomov \cite{Olshanetsky1976} considered the Hamiltonian 
\begin{align} \label{Hamgen}
    H(q,p) = \sum_{j=1}^n \frac{p^2_i}{2} - \sum_{\alpha \in \mathcal{R} (\mathfrak{g})} m^2_{ | \alpha | } \wp (\alpha . q) 
\end{align}
where $n$ is the rank of the associated Lie algebra and $\mathcal{R} (\mathfrak{g})$ denotes its set of roots, $m_{ | \alpha | }$ are mass parameters dependent only on the length of the root $\alpha$ to preserve the invariance of the Hamiltonian system under the action of the Weyl group. \\

Before solving the zero curvature equation of the associated Lax pair defined in the previous section, one needs to make a choice of the functions $x(u,z)$ and $y(u,z)$, several choices exist in the literature since the conditions set on both of these functions leave a range of possibilities and the Hamiltonian structure is independent of this choice. We shall make the following choice (\cite{Takasaki_1999})
\begin{align} \label{choice}
 x(u,z) = \frac{\theta_1 (z-u) \theta'_1 (0)}{\theta_1(z) \theta_1 (u)}   
\end{align}
where the function $\theta_1(u)$ is a Jacobi elliptic theta function defined by
\begin{align}
    \theta_1 (u) = - \sum_{n=-\infty} ^{\infty}  e^{ \left( \pi i \tau \left( n+\frac{1}{2} \right)^2 +2 \pi i \left( n+\frac{1}{2} \right) \left( u +\frac{1}{2} \right)  \right)} 
\end{align}
This choice fixes $y(u,z)$
\begin{align}
    y(u,z) = - x(u,z) (\rho(u)+ \rho(z-u) )
\end{align}
where $\rho(u)$ is the logarithmic derivative of the theta function. These functions enjoy a lot of properties (discussed in Appendix \ref{PProperties}), and its due to these properties that our system fits as an isomnodromy system. Indeed, take the Lax pair defined in (\ref{Laxpair}), from the properties of the choice taken one may write
\begin{align}
    \td{L}(z) &\overset{z\to 0}{=} - i g \sum_{j\neq k} E_{jk} \left( \frac{1}{z}  + \rho(q_j-q_k) + O(z)\right) + P \cr
    &\overset{z\to 0}{=} \td{L}^{[0,1]}z^{-1} + \td{L}^{[0,0]} 
\end{align}
which has a leading polar part  
\begin{align}
    \td{L}^{[0,1]} = - i g \sum_{j\neq k} E_{jk} \frac{1}{z}
\end{align}
This leading order is always diagonalizable, nevertheless, admits eigenvalues with a higher multiplicity which indicates that we are in a regular simple resonant case. To prove that, one solves the eigenvalue problem of the matrix and gets the gauge matrix whose columns are the respective eigenvectors. The gauge matrix is given by 
\begin{align}
    J:= & \begin{pmatrix}
        -1 & -1 & \hdots & -1 & 1 \\
        0 &   0 & \hdots & 1 & 1  \\
        \vdots &   \vdots  &  \mathbf{1}   & \vdots  & \vdots \\
        0 & 1 & \hdots & 0 & 1 \\
        1 & 0 & \hdots & 0 & 1 
    \end{pmatrix}
\end{align}
and the eigenvalues are given in the following column vector
\begin{align}
    V:= \begin{pmatrix}
        -1 \\
        \vdots \\
        -1 \\
        n-1
    \end{pmatrix}
\end{align}
Indeed, one may already deduce the existence of the diagonal form from the set of eigenvalues. This construction fits smoothly into the geometric picture, and we refer to the element of the irregular type as the isomonodromic time, in our setting it is $-ig$. Thus, isomonodromic deformations are performed keeping the parameter $g$ fixed. Furthermore, as we shall see, the deformation parameters are reduced to the module of the torus $\mathbf{t} = \tau$. \\

As explained in the previous section, the gauge in which the Lax pair $(\td{L}(z),\td{A}(z))$ is only quasi-periodic satisfying (\ref{action}). Let us use the gauge action defined in (\ref{Gx}) in terms of the Lamé functions to get a periodic Lax pair. This gauge is obtained by 
\begin{align}
    \td{L}(x) = G_{l}(z) L(z) G_{l}^{-1}(z) + G_{l}'(z) G_{l}^{-1}(z)
\end{align}
Acting by the inverse matrix provides a Lax matrix given by
\begin{align}
    L(z) = & G_{l}^{-1}(z) \td{L}(z) G_{l}(z) - G_{l}^{-1}(z) G_{l}'(z) \cr
    = & P + ig \sum_{j \neq k} \left( \frac{x(q_j-q_k , z) x(q_k,z) }{x(q_j,z)} \right) E_{jk} - \sum_j \left(  \frac{\partial_z x(q_j , z)}{x(q_j , z)} \right)E_{jj}
\end{align}
The Lax matrix preserves its periodicity around the $\mathcal{A}$ cycle. The quasi-periodic behavior of the function $x(u,z)$ is clearly preserved by the derivative, thus, the first and last terms are indeed periodic. As for the second term, we have
\begin{align}
    \left( \frac{x(q_j-q_k , z+ \tau) x(q_k,z+\tau) }{x(q_j,z+\tau)} \right) E_{jk} = & \left(  \frac{ e^{2 \pi i (q_j-q_k)} x(q_j-q_k , z)  e^{2 \pi i q_k} x(q_k,z) }{e^{2 \pi i q_j} x(q_j,z)} \right) E_{jk}
 \end{align}
which is indeed periodic. The auxiliary matrix in this gauge is given by
\begin{align}
    A(z) = & G_{l}^{-1}(z) \td{A}(z) G_{l}(z) - 2 \pi i G_{l}^{-1}(z) \partial_\tau G_{l}(z) \cr
    = & D + ig \sum_{ j \neq k} \left( \frac{y(q_j-q_k , z) x(q_k,z) }{x(q_j,z)} \right) E_{jk} - 2 \pi i \sum_j  \left( \frac{\partial_\tau x(q_j,z) }{x(q_j,z)}  \right) E_{jj}
\end{align}
which gives the periodicity of the auxiliary matrix. Note that despite having a periodic Lax matrix, the polar structure becomes un-natural due to the meromorphic property of the Lamé functions. Both matrices $L(z)$ and $\td{L}(z)$ admit a local description that fits the geometric construction. However, the price to pay when requiring a periodic behavior (thus compensating the twists in (\ref{action})) is the introduction of an additional set of singularities. The local behavior of the Lax matrix remains independent of both gauges, this is due to the commutation of diagonal matrices. 

\begin{remark}
    Another choice widely used in the literature is in terms of the Weierstrass sigma function $\sigma(u  | 1,\tau)$ with primitive periods $1$ and $\tau$ 
    \begin{align}
        x(u,z) = \frac{\sigma(z-u)}{\sigma(z) \sigma(u)}
    \end{align}
    Note that this choice is irrelevant when it comes to the Hamiltonian structure of the system. The two choices are actually related due to the well-known relation between the elliptic Weierstrass $\sigma$ function and the Jacobi theta functions $\theta_1$ in terms of the periods of the lattice. 
\end{remark}
Now that one has the explicit expressions of the Lax matrices, one solves the zero curvature equation in order to obtain the evolution of the Darboux coordinates. the symplectic structure holds no dependence on the gauge considered, hence both gauges provide the same evolution. From these evolutions, one obtains the Hamiltonian of the system, we have 
\begin{theorem}
The zero curvature equation gives the evolution of the Darboux coordinates 
\begin{align}
    2 \pi i \partial_\tau q_j = & p_j \cr
    2 \pi i \partial_\tau p_j = & - g^2 \sum_{k \neq j} \wp'(q_k-q_j,\tau),
\end{align}
In particular, these evolutions are Hamiltonian meaning 
\begin{align}
    2 \pi i \partial_\tau q_j = \{ q_j,H(q,p,\tau) \} && 2 \pi i \partial_\tau p_j = \{ p_j,H(q,p,\tau) \}
\end{align}
with the Hamiltonian given by
\begin{align} \label{Ham}
    H (q,p,\tau) = \frac{1}{2} \sum_{i=1}^n p_j^2 + \frac{g^2}{2} \sum_{k \neq j} \wp (q_k-q_j,\tau) 
\end{align}
Where the bracket employed is the canonical Poisson bracket 
\begin{align}
    \{ q_j,p_k \} = \delta_{jk} \, \, \, \, \, \, \, \, \,   \{ p_j,p_k \} =  \{ q_j,q_k \} =0
\end{align}
One realizes the Hamiltonian of the elliptic Calogero-Moser system, or $A_{n-1}$, built from the isomonodromic deformation of a Lax system defined over the torus $T_\tau$ with one marked point located at $z=0$. Furthermore, this matrix admits a local factorization in the neighborhood of the pole giving an irregular type and a residue. 
\end{theorem}
\begin{proposition}
    The Hamiltonian system (\ref{Ham}) is the general Hamiltonian system defined in (\ref{Hamgen}) for the Lie algebra $\mathfrak{su}(n)$.  Moreover, the system is invariant under the action of the Weyl group.
\end{proposition}
\begin{proof}
    The roots of $\mathfrak{su}(n)$ lie on a hyperplane $\mathbb{C}^n$ which motivates the use of $(q_i,p_i)_{1 \leq i \leq n}$ as dynamical variables and the number of degrees of freedom is $n-1$ which is equivalent to the rank of $\mathfrak{su}(n)$. Moreover, the roots are given by $\alpha = (q_i-q_j) _{1 \leq i,j \leq n}$ for $i \neq j$. One realises after that the system (\ref{Ham}) as a special case of the more general Hamiltonian (\ref{Hamgen}). On the other side, it is well known that the Weyl group in this case is the symmetric group $S_n$ of permutations with $n$ variables, it is clear due to the sums that the Hamiltonian is invariant under permutations of the coordinates $p_i \rightarrow p_{\mathcal{\sigma} (i)}, \, \, \, q_i \rightarrow q_{\mathcal{\sigma} (i)}$ with $\mathcal{\sigma} \in S_n$.  
\end{proof}
    
\begin{remark}
    By exchanging $2 \pi i \partial_\tau \rightarrow \partial_t$ in the construction, one gets the same Hamiltonian and the zero curvature equation 
    \begin{align}
        \partial_t \td{L}(z) + \partial_z \td{A}(z) = [\td{L}(z),\td{A}(z)]
    \end{align}
    which is the usual isospectral Lax equation studied to get the $A_{n-1}$ system. Note that following Manin's convention, one sets the primitive periods to $2 \omega_1 = 1$ and $2 \omega_3 = \tau$.
\end{remark}
\begin{remark}
    The elliptic CM system has been studied extensively from the point of view of isospectral integrable systems, In particular, a set of conserved quantities may be extracted from the trace of the powers of the Lax matrix, this set forms a completely integrable structure. Evidently, the Hamiltonian itself is related to the quadratic trace of the Lax matrix. In this construction, one usually takes the choice of $x(u,z)$ and $y(u,z)$ as functions of the Weierstrass sigma function. Note that this choice is not suitable however for the construction of isomonodromy systems.
\end{remark}
\begin{remark}
    The setup in \cite{Takasaki_1999} relating the isospectral approach to the isomonodromic one by a simple exchange works exclusively in this setting, indeed, both approaches admit huge similarities in the case of one simple pole. Nevertheless, this is not a general fact as proven in \cite{Marchal_2024} for a case on the Riemann sphere. 
\end{remark}

\subsection{The fundamental symplectic $2-$form and the symplectic structure}
The non-autonomous nature of the Hamiltonian system motivates the extension of the phase space, we define the symplectic $2-$form for a non-autonomous Hamiltonian  
\begin{definition}
We define the Fundamental symplectic $2-$form
\begin{align}
    \Omega_{iso} = \sum_{j=1}^n dq_j \wedge dp_j - \frac{1}{2 \pi i } d H(q,p,\tau) \wedge d\tau
\end{align}
as a fibre bundle above the base, it is a $2-$form on the phase space with coordinates $(p_j,q_j,\tau)$. 
\end{definition}

\begin{remark}
    In \cite{Del_Monte_2023}, the authors considered a more general symplectic structure where they included the monodromy in the term $d g \wedge d\eta $ when defining the symplectic $2-$form. This is indeed possible since the Hamiltonian admits an explicit dependence on the monodromy. The term $\eta$ is an additional degree of freedom related to the diagonalization of the monodromy dependence of the Lax matrix. Note that in this article the monodromy is considered fixed and hence this term does not appear.   
\end{remark}

The defined extended symplectic $2-$form clearly admits a restriction to each fibre equals to that of the canonical $2-$form associated to the system. The existence of the canonical $2-$form is well known due to Darboux's theorem and the existence of its extension is discussed in remark \ref{extension}. Let us review how from this definition, one may define the connection form and prove its symplectic property.   \\

The dynamics of the system on the extended phase space and its extended set of coordinates may be encoded in a Hamiltonian vector that presents the evolution of the coordinates, we define this vector field by contracting the extended $2-$form
\begin{definition}
    Let us define the horizontal vector field as
    \begin{align}
        X_H = 2 \pi i  \frac{\partial}{\partial \tau} + \sum_{j=1}^n \left( \frac{\partial H}{\partial p_j} \frac{\partial}{\partial q_j } - \frac{\partial H}{\partial q_j} \frac{\partial}{\partial p_j }    \right)
    \end{align}
    The definition of this vector field is obtained by solving the equation $i_{X_H} \Omega = 0$, where $i_{X_H}$ is the interior product with respect to $X_H$. In particular, the Ehresmann connection spanned by $X_H$ defining the Hamiltonian system is seen as the $\Omega$ orthogonal complement of the vertical bundle 
\end{definition}
It is crucial to notice that the vector field $X_H$ spans the space of Horizontal vector fields, i.e. the vectors that encode the changes of the Darboux coordinates induced by the change of the third coordinate $\tau$ (The vertical vectors on the other side spans the space where $\tau$ is held constant). The concept of a connection on a fibered manifold above a given manifold is a smooth field of subspaces transversal to the fibres in the sense of Ehresmann. Any vector tangent to the fibered space may be decomposed into a vertical vector and a horizontal one. Any curve $\gamma$ on the base starting at a point $a_0$ and ending at a point $a_1$ may be lifted to a curve on the fibered space starting at a point $\overline{a}_0$ above $a_0$ and ending at the point $\overline{a}_1$ above $a_1$. This is the well known horizontal lift of the base curve $\gamma$ and it realizes the parallel transport of $\overline{a}_0$. 

\begin{definition}
The connection form is defined as the operator that annihilates the horizontal vector field 
\begin{align}
    \theta (X_H) = 0
\end{align}
In particular, this definition ensures the flatness of the connection since the vector field $X_H$ is chosen so that $i_{X_H}\Omega = 0$ or in terms of the connection as $\theta = Ker \Omega$.
\end{definition}

\begin{remark}
    There are three equivalent ways to define a connection on a fibered space to encode the same notion of parallel transport, the one we are using is based on the use of generators of the vector fields. Other equivalent ways are based on the use of the derivatives of the Darboux coordinates and on differential forms.    
\end{remark}

This definition and the properties of the fundamental $2-$from state that the Hamiltonian system is in fact completely determined by the symplectic $2-$form. Moreover, for any $2-$form defined on the total space of a symplectic fibre bundle admitting a restriction on each fibre equal to that of the canonical $2-$form, the $\Omega$ orthogonal complement of the vertical bundle defines an Ehresmann connection on the tangent space. It is a well known fact that if $\Omega$ is closed, then the induced connection is symplectic. 

\begin{remark} \label{extension}
     The problem of finding a closed extension of the canonical $2-$form along the fibres is a purely topological problem that depends on the de Rham cohomology class, this problem was solved in \cite{article} with a hypothesis on the fibre. The proof of the existence admits a fruitful interpretation in terms of connections and their parallel transport \cite{article,Guillemin1996SymplecticFA}. 
\end{remark}

\begin{proposition}
    The fundamental symplectic $2-$form is closed, and the connection form defined is symplectic i.e. preserves the symplectic structure of the system.
\end{proposition}
\begin{proof}
    let us take the exterior derivative of the $2-$form defined
    \begin{align}
        d \Omega_{iso} = & d \left( \sum_{j=1}^n dq_j \wedge dp_j  \right) - \frac{1}{2 \pi i } d \left(  d H(q,p,\tau) \wedge d\tau \right)
    \end{align}
    The first term is the canonical $2-$form on each fibre and is known to be exact, hence vanishes We are left with the second term which we write
    \begin{align}
        d \left(  d H(q,p,\tau) \wedge d\tau \right) =& d \left( \sum_{j=1}^n \frac{\partial H}{\partial q_j } d q_j \wedge d \tau    +  \sum_{j=1}^n \frac{\partial H}{\partial p_j } d p_j \wedge d \tau + \frac{\partial H}{\partial \tau } d \tau \wedge d \tau \right) 
    \end{align}
    The term $d \tau \wedge d \tau$ vanishes leaving us with 
    \begin{align}
         d \left(  d H(q,p,\tau) \wedge d\tau \right) =& \sum_{j=1}^n d \left( \frac{\partial H}{\partial q_j} \right) \wedge dq_j \wedge d \tau + \sum_{j=1}^n d \left( \frac{\partial H}{\partial p_j} \right) \wedge dp_j \wedge d \tau  \cr
         &  + \sum_{j=1}^n \frac{\partial H}{\partial q_j} d \left( dq_j \wedge d\tau \right) + \sum_{j=1}^n \frac{\partial H}{\partial p_j} d \left( dp_j \wedge d\tau \right) \cr
         = & \sum_{j=1}^n \left( \frac{\partial^2 H}{\partial \tau \partial q_j} d \tau + \sum_{k=1}^n \frac{\partial^2 H}{\partial q_k \partial q_j} d q_k + \sum_{k=1}^n \frac{\partial^2 H}{\partial p_k \partial q_j} d p_k   \right) \wedge d q_j \wedge d \tau  \cr
        & + \sum_{j=1}^n \left( \frac{\partial^2 H}{\partial \tau \partial p_j} d \tau + \sum_{k=1}^n \frac{\partial^2 H}{\partial q_k \partial p_j} d q_k + \sum_{k=1}^n \frac{\partial^2 H}{\partial p_k \partial p_j} d p_k   \right) \wedge d p_j \wedge d \tau \cr
    \end{align}
    Due to the anti-symmetric property of the wedge product and that the second derivatives of the Hamiltonian commute, one obtains that $d \Omega_{iso} =0$, which ends the proof.
\end{proof}
\begin{remark}
    The last proposition states the closure of the extension of the canonical $2-$form, which induces the symplectic property of the connection. In other words, parallel transport of connections spanned by the vector field $X_H$ preserves the symplectic structure.   
\end{remark}

\section{Conclusion and outlooks}
In this article, we have revisited the theory of isomonodromic deformations on the torus and discussed the elliptic CM model which appears as an instance of the general theory. The non-autonomous nature of the Hamiltonian system motivated the consideration of the extended symplectic $2-$ form. Several generalizations of the present work are possible, notably; 
\begin{enumerate}
    \item The isomonodromy systems of the Painlevé equations is a well known field of studies, the generalization of the construction of isomonodromy systems in the elliptic case has been studied by many authors. Following \cite{marchal2024hamiltonianrepresentationisomonodromicdeformations,MarchalAlameddineP1Hierarchy2023}, one would consider the moduli space of meromorphic connections having a periodic pole structure and benefit from the gauge action to solve the zero curvature equation that yields the evolution of the Darboux coordinates in a more general setting. The periodicity of the pole structure shall be encoded in an elliptic function that appears in the entries of the Lax matrix as the case of the elliptic CM systems. The local fundamental decomposition is also available in this setting, however, one must understand the twist above the torus and how the monodromy is considered. After getting an explicit candidate for the Hamiltonian, one performs the reduction using the symmetries of the tangent space and gets a generalized version of the elliptic CM system up to a symplectic change of coordinates. Thus obtaining the resulted Hamiltonian as an isomonodromic deformation system of a meromotphic connection above the torus. This in particular builds the bridge between the rational world (performed on the sphere) and the elliptic world of isomonodromic deformations and their Hamiltonian systems. This also requires a deep understanding of isomonodromy systems and a suitable definition of an \textit{absolute connection} in the sense of \cite{Boalch2012}.
    \item A natural question that arises when studying isomonodromy systems is that of the \textit{quantization} of these systems and the resultant quantum version. It is now very well known that isomonodromic systems on the sphere give rise to the KZ equation, which is an equation for correlation functions in the Wess–Zumino–Witten model for two-dimensional conformal field theory \cite{KNIZHNIK198483}. This system was obtained by deformation quantization of the Schlesinger system \cite{Beauville_1994,Las98} and the KZ connection is known to be equivalent to a genus zero Hitchin connection in geometric quantization \cite{harnad1994quantumisomonodromicdeformationsknizhnikzamolodchikov,Reshetikhin1992}. Several generalizations followed \cite{millson2004casimiroperatorsmonodromyrepresentations,felder2000differentialequationscompatiblekz,Rembado_2019}. The quantization of isomonodromic deformations on the torus have been shown to give KZB equations \cite{Korotkin_1997}, one possible generalization of this result may be achieved by using the \textit{topological recursion} \cite{EO07} and its quantization \cite{Quantization_2021}. Indeed, in this setting one needs to consider $\hbar-$connections through the introduction of a formal parameter $\hbar$ using a suitable rescaling.
    \item The JMU differential $\omega_{JMU}$ has been computed in a more extended $SL_n(\mathbb{C})$ setting in \cite{Del_Monte_2023}, through which the isomonodromic tau function was obtained for isomonodromic deformations on the torus. The construction of a random matrix model allows the correspondence between the isomonodromic tau function and the partition function associated to the matrix model defined by the general potential up to some constant terms dependent on the normalization. One would expect an identification between both functions, by analogy to what happens in the rational case. The geometric construction of isomonodromy systems on the torus is somewhat analogous to the one over the sphere. On the contrary, the construction of a matrix model in the elliptic setting admits substantial changes and modifications.   
\end{enumerate}

\section*{Acknowledgments}
The author is grateful to his supervisor Olivier Marchal for fruitful discussions and remarks on the manuscript of the paper. The author thanks Gabriele Rembado, Mathieu Dutour, Mattia Cafasso for discussions and Fabrizio Del Monte for discussions and for the reference \cite{Bonelli2021}. 

\appendix
\section{Reminder on elliptic functions} \label{B}
\subsection{The Theta function and its product representation}
The properties of the Lamé function $x(u,z)$ defined in (\ref{choice}) are dependent on the properties of the elliptic Jacobi function. In this appendix, we recall some elementary properties for the Jacobi theta function defined and its product representation. Let us write 
\begin{align}
    \theta_1(z) = - \sum_{n=-\infty}^\infty \nu^{(n + \frac{1}{2})^2} \eta^{2(n+\frac{1}{2})}
\end{align}
where $\nu = exp(\pi  i \tau)$ is the nome and $\eta = exp( \pi i (z + \frac{1}{2}))$. Due to Liouville's theorem, it is  a quasi-periodic function with 
\begin{align}
    \theta_1 (z+1) = & - \theta_1(z) \cr
    \theta_1 (z+\tau) = & -e^{- \pi i (\tau + 2z )} \theta_1(z) 
\end{align}
It is an entire function with simple zeros located at the lattice points $z=m+n\tau$. This expression is related to the famous Jacobi theta function via
\begin{align}
    \theta_1(z) = & - e^{\frac{1}{4} \pi i \tau + \pi i z} \upsilon (z+ \frac{1}{2} + \frac{1}{2} \tau ; \tau) \cr
    = & - e^{\frac{1}{4} \pi i \tau + \pi i z} \sum_{n=-\infty}^\infty (e^{\pi i n^2 \tau}) e^{2 \pi i n (z +\frac{1}{2} + \frac{1}{2} \tau)}
\end{align}
In order to simplify the form of the Lamé function, we resort to the product representation of theta functions, indeed \cite{hardy75}, the presentation of the theta function in terms of the nome and the argument allows an application of the Jacobi triple product which is a special case of the Macdonald identities. This product is written
\begin{align}
    \sum_{n=\infty}^\infty \nu^{n^2} \eta^{2n} = \prod_{m=1}^\infty \left( 1- \nu^{2m} \right) \left( 1+  \eta^2 \nu^{2m-1}  \right) \left( 1 +  \eta^{-2}  \nu^{2m-1}  \right)
\end{align}
This equality holds for any complex numbers $\nu$ and $\eta$ with the assumptions  $ |\nu|  < 1$ and $\eta \neq 0$. Applying this equality to $\theta_1(z)$ expressed in terms of the Jacobi elliptic fnunction provides
\begin{align}
    \theta_1(z) = & i \nu^{\frac{1}{4}} \eta \sum_{n=- \infty}^\infty q^{n^2} \td{\eta}^{2n} \cr
    = & i \nu^{\frac{1}{4}} \eta  \prod_{m=1}^\infty \left( 1- \nu^{2m} \right) \left( 1+  \td{\eta}^2 \nu^{2m-1}  \right) \left( 1 + \td{ \eta}^{-2}  \nu^{2m-1}  \right) \cr
    = & i \nu^{\frac{1}{4}}  \prod_{m=1}^\infty \left( 1- \nu^{2m} \right) \left( 1+  \eta^2 \nu^{2m}  \right) \left( 1 +  \eta^{-1}  \nu^{2m-2}  \right)
\end{align}
where we used $\td{\eta} = exp(\pi i (z +\frac{1}{2}+ \frac{1}{2} \tau)) = \eta \nu^{\frac{1}{2}}$  . This expression admits the following derivative
\begin{align}
     \theta'_1(z) = & i \nu^{\frac{1}{4}} \sum_{j=1}^\infty \bigg(  \left( 1- \nu^{2j} \right) \left( 2 \pi i  \eta^2 \nu^{2j} \left( 1 +  \eta^{-1}  \nu^{2j-2}  \right) - \pi i \eta^{-1}  \nu^{2j-2}  \left( 1+ \eta^2 \nu^{2j} \right) \right) \cr
     & \left( \prod_{m=1, m \neq j}^\infty \left( 1- \nu^{2m} \right) \left( 1+  \eta^2 \nu^{2m}  \right) \left( 1 +  \eta^{-1}  \nu^{2m-2}  \right) \right) \bigg)
\end{align} 
This expression may also be viewed as 
\begin{align}
    \theta'_1(z) = & \theta_1(z) \sum_{j=1}^\infty \frac{ 2 \pi i \eta^2 \nu^{2j}  \left( 1+ \eta^{-1} \nu^{2j -2} \right) - \pi i \eta^{-1}  \nu^{2j-2}  \left( 1+ \eta^2 \nu^{2j} \right)   }{\left( 1+  \eta^2 \nu^{2j}  \right) \left( 1 +  \eta^{-1}  \nu^{2j-2}  \right)}
\end{align}
evaluated at $z \to 0$ both expressions lead to the forms
\begin{align}
        \theta'_1(0) = & i \nu^{\frac{1}{4}} \sum_{j=1}^\infty \bigg(  \left( 1- \nu^{2j} \right) \left(- 2 \pi i   \nu^{2j} \left( 1 - i  \nu^{2j-2}  \right) - \pi \nu^{2j-2} \left( 1 - \nu^{2j} \right) \right) \cr
     & \left( \prod_{m=1, m \neq j}^\infty \left( 1- \nu^{2m} \right) \left( 1 - \nu^{2m}  \right) \left( 1 -i  \nu^{2m-2}  \right) \right) \bigg) \cr
     = &    \pi  \nu^{\frac{1}{4}} \sum_{j=1}^\infty  \bigg(   \left( 1- \nu^{2j} \right) \left( -2 i \nu^{2j} -  \nu^{4j-2} - \nu^{2j-2}\right) \cr
     & \left( \prod_{m=1, m \neq j}^\infty \left( 1- \nu^{2m} \right) \left( 1 - \nu^{2m}  \right) \left( 1 -i  \nu^{2m-2}  \right) \right) \bigg) 
\end{align}
or, equivalently
\begin{align}
     \theta'_1(0) = &  \theta_1(0) \sum_{j=-1}^\infty \left(  \frac{-2 \pi i \nu^{2j}}{1-\nu^{2j}}  +\frac{ - \pi \nu^{2j-2}}{1-i \nu^{2j-2}} \right) \cr
     = &  i \nu^{\frac{1}{4}} \prod_{m=1}^\infty \left( 1-\nu^{2m} \right) \left( 1-\nu^{2m} \right) \left( 1-i \nu^{2m-2} \right)  \sum_{j=1}^\infty \left(  \frac{-2 \pi i \nu^{2j}}{1-\nu^{2j}}+   \frac{- \pi \nu^{2j-2}}{1-i \nu^{2j-2}} \right)
\end{align}
This expression has another interpretation particularly interesting when $z$ does not have an imaginary part
\begin{align}
    \theta_{1}(z) = - 2 \nu^{\frac{1}{4}} \cos(\pi (z+\frac{1}{2})) \prod_{m=1}^\infty \left(1 - \nu^{2m} \right) \left(1 + 2 \cos(2 \pi (z+\frac{1}{2}) )\nu^{2m} + \nu^{4m} \right)
\end{align}
In order to use this form, one makes the assumption that the nome $ |\nu|  < 1$. This form admits the derivative 
\begin{align}
    \theta_1' (z) = & 2 \pi \nu^{\frac{1}{4}} \sin(\pi (z+\frac{1}{2})) \prod_{m=1}^\infty \left(1 - \nu^{2m} \right) \left(1 + 2 \cos(2 \pi (z+\frac{1}{2}) )\nu^{2m} + \nu^{4m} \right) \cr
    & - 2 \nu^{\frac{1}{4}} \cos(\pi (z+\frac{1}{2})) \sum_{j=1}^\infty \bigg( \left( \left(1 - \nu^{2j} \right) \left(-4 \pi \sin(2 \pi (z+\frac{1}{2}) ) \nu^{2m}  \right) \right) \cr
    & \prod_{m=1, \,\, m \neq j }^\infty  \left(1 - \nu^{2m} \right) \left(1 + 2 \cos(2 \pi (z+\frac{1}{2}) )\nu^{2m} + \nu^{4m} \right) \bigg)
\end{align}
The above expression at the limit when $z\to 0$ simplifies to the expression
\begin{align}
     \theta_1' (0) = & 2 \pi \nu^{\frac{1}{4}} \prod_{m=1}^\infty \left(1 - \nu^{2m} \right) \left( 1-2 \nu^{2m} + \nu^{4m}  \right)
\end{align}
Although defined slightly different then the famous Jacobi elliptic theta functions, $\theta_1(z)$ satisfies the well known heat equation
\begin{align} \label{heat}
    4 \pi i \frac{\partial \theta_1 (z)}{\partial \tau} = \frac{\partial^2 \theta_1(z)}{\partial z ^2}
\end{align}
of the Jacobi theta functions. 

\subsection{The Lamé functions} \label{PProperties}
The choice of functions defined in (\ref{choice}) admits a very rich structure and the quasi periodicity needed is ensured by these properties, among them are
\begin{proposition}
    The Lamé function defined admits the following analytical properties;
    \begin{itemize}
        \item The zeroes of the Jacobi theta function turns $x(u,z)$ into a meromorphic function in both variables with poles located at the lattice points $u=m+n \tau$ and $z=m+n \tau$.
        \item From the quasi periodic behavior of the Jacobi theta function, $x(u,z)$ inherits the following quasi periodic behavior
        \begin{align}
    x(u+1,z) = x(u,z) && x(u+ \tau,z) = e^{2 \pi i z} x(u,z) \cr
    x(u,z+1) = x(u,z) && x(u, z + \tau) = e^{2 \pi i u} x(u,z).
        \end{align}
       \item At the origin, the function $x(u,z)$ admits the following analytic behavior
        \begin{align}
       x(u,z) &\overset{z\to 0}{=} - \frac{1}{z} + \rho(u) + O(z)  \cr
       x(u,z) &\overset{u\to 0}{=}  \frac{1}{u} - \rho(z) + O(u)  
        \end{align}
        this property is essential for the link with the geometric construction. 
        \item The Lamé function satisfies the equation
        \begin{align}
            2 \pi i \frac{\partial x(u,z)}{\partial \tau } + \frac{\partial^2 x(u,z)}{\partial u  \partial z} = 0
        \end{align}
        this is a consequence of the heat equation (\ref{heat}) satisfied by the theta functions.
    \end{itemize}
\end{proposition}
To prove these properties one relies on the properties of the Jacobi theta function discussed in the previous subsection of the appendix. One may extract as well similar properties for the function $\rho(z)$;
\begin{proposition}
The function $\rho(z)$ enjoys the following set of analytic properties
\begin{itemize}
    \item It is a meromorphic function with poles located at the Lattice points $z=m+n \tau$.
    \item It is an odd function with an additive quasi periodicity
    \begin{align}
        \rho(-z) = - \rho(z), && \rho(z+1) = \rho(z), && \rho(z+\tau) = \rho(z)- 2 \pi i 
    \end{align}
    easily seen from the behavior of the theta functions.
    \item At the origin of the complex plane, it admits the following behavior
    \begin{align}
        \rho(z)  &\overset{z\to 0}{=} \frac{1}{z} + \frac{\theta_1'''(0) }{3 \theta_1' (0)}z + O(u^3).
    \end{align}
\end{itemize}
\end{proposition}

\addcontentsline{toc}{section}{References}
\bibliographystyle{plain}
\bibliography{Biblio}
\end{document}